\documentclass[lettersize,journal]{IEEEtran}
\usepackage{amsmath}
\usepackage{amsfonts}
\usepackage{amssymb}
\usepackage{amsthm}
\usepackage{tabularx}
\usepackage{mathrsfs} 

\usepackage{url}
\usepackage{hyperref}

\newtheorem{Proposition}{Proposition}

\usepackage{stfloats}
\usepackage{float}
\usepackage{graphicx}
\usepackage{cite}
\usepackage{xcolor}
\usepackage{subfigure}

\makeatletter
\def\blfootnote{\xdef\@thefnmark{}\@footnotetext}
\makeatother

\begin{document}
	
		\title{\huge{Performance Analysis of RIS/STAR-IOS-aided V2V NOMA/OMA Communications over Composite Fading Channels}} 
\author{Farshad~Rostami~Ghadi\IEEEmembership{},~Masoud~Kaveh,~and Diego~Mart\'in}%,~\IEEEmembership{Fellow}, \textit{IEEE}
%}
	\maketitle
	
	\begin{abstract}
This paper investigates the performance of vehicle-to-vehicle (V2V) communications assisted by a reconfigurable intelligent surface (RIS) and a simultaneous transmitting and reflecting intelligent omni-surface (STAR-IOS) under non-orthogonal multiple access (NOMA) and orthogonal multiple access (OMA) schemes. In particular, we consider that the RIS is close to the transmitter vehicle while the STAR-IOS is near the receiver vehicles. In addition, we assume that the STAR-IOS exploits the energy-splitting (ES) protocol for communication and the fading channels between the RIS and STAR-IOS follow composite Fisher-Snedecor $\mathcal{F}$ distribution. Under such assumptions, we first use the central limit theorem (CLT) to derive the PDF and the CDF of equivalent channels at receiver vehicles, and then, we derive the closed-form expression of outage probability (OP) under NOMA/OMA scenarios. Additionally, by exploiting Jensen's inequality, we propose an upper bound of the ergodic capacity (EC), and then, we derive an analytical expression of the energy efficiency (EE) for both NOMA and OMA cases. Further, our analytical results, which are double-checked with the Monte-Carlo simulation, reveal that applying RIS/STAR-RIS in V2V communications can significantly improve the performance of intelligent transportation systems (ITS). Besides, the results indicate that considering the NOMA scheme provides better performance in terms of the OP, EC, and EE as compared with the OMA case for the considered V2V communication.   
	\end{abstract}
	%\vspace{-3.5ex}
	\maketitle
	\blfootnote{\noindent %Manuscript received XX May, 2023. 
		%This work was funded in part by Junta de Andaluc\'ia through grant EMERGIA20-00297, and in part by MCIN/AEI/10.13039/501100011033 through grant PID2020-118139RB-I00. %The review of this paper was coordinated by XXXX.
	}
	
	\blfootnote{\noindent Farshad Rostami Ghadi is with the Communications and Signal Processing Lab, Telecommunication Research Institute (TELMA), Universidad de M\'alaga, M\'alaga, 29010, Spain (e-mail: $\rm farshad@ic.uma.es$).}
		\blfootnote{\noindent Masoud Kaveh is with the Department of Information and Communication Engineering, Aalto University, Espoo, Finland (e-mail: $\rm masoud.kaveh@aalto.fi$).}
	\blfootnote{\noindent Diego Mart\'in is with the ETSI de Telecomunicaci\'on, Universidad Polit\'ecnica de Madrid, 28040 Madrid, Spain (e-mail: diego.martin.de.andres@upm.es).
			
		\thanks{Corresponding author: Diego Mart\'in.}
	
	\thanks{Digital Object Identifier 10.1109/XXX.2021.XXXXXXX}}
	%\blfootnote{\noindent Copyright (c) 2015 IEEE. Personal use of this material is permitted. However, permission to use this material for any other purposes must be obtained from the IEEE by sending a request to pubs-permissions@ieee.org.} 
%	\blfootnote{Manuscript received January 25, 2021; revised XXX. The review of this paper was coordinated by XXXX.} 
%		\blfootnote{This work has been funded in ..}
%	 	\blfootnote{\noindent The authors are with the .. (e-mail: $\rm$.}
	
%	\blfootnote{Digital Object Identifier 10.1109/XXX.2021.XXXXXXX}
	%\IEEEpeerreviewmaketitle
	\vspace{-0.5cm}
	\section{Introduction}\label{sec-intro}
Developments in the context of intelligent transportation systems (ITSs) have led to significant attention being paid to vehicular communication in recent years. Vehicular communications, also known as vehicle-to-everything (V2X) communications, are promising technology to support cutting-edge applications, such as self-driving cars which will be provided by the sixth-generation (6G) wireless technologies \cite{wang2021green}. Furthermore, V2X communications, e.g., vehicle-to-vehicle (V2V) communications, will provide remarkable advantages such as enhancing the quality-of-service (QoS) requirements of the vehicle user equipment, improving road safety and air quality, and increasing the cooperation awareness of all vehicles in the local driving environment \cite{noor20226g}. However, to achieve the above-mentioned ambitious goals in future wireless technologies such as 6G, it will be needed the integration of a range of technologies, including unmanned aerial vehicles (UAVs), reconfigurable intelligent surfaces (RISs), visible light communication (VLC), edge caching, non-orthogonal multiple access (NOMA), etc. Therefore, by combining such key technologies, future 6G-V2X networks will be able to provide an intelligent, autonomous, user-driven connectivity and service platform for ITSs. 

RIS has been recently introduced as a promising technology due to its capability to extend coverage area and provide high spectral/energy efficiency for 6G wireless communications \cite{basar2019wireless}. Generally speaking, RIS is an artificial metasurface with a large number of low-cost passive reflecting elements that can intelligently and adaptively enable a smart radio environment (SRE) and maximize the desired signal quality at receivers. However, one of the key limitations of RIS is that the transmitter and receiver need to be located on the same side of the RIS, meaning that the RIS can only support the half-space of the SRE. To address this challenge, a novel concept of simultaneous transmitting and reflecting intelligent omni-surface (STAR-IOS) has been proposed which can enable a 360$^{\circ}$ coverage region of SREs \cite{liu2021star}. By exploiting this unique feature of STAR-IOS, separate reflection and refraction passive beamforming can be designed for the various serving regions. Moreover, given the superiority of STAR-IOS in providing full-space coverage area, there has been recently a growing interest in integrating these devices with NOMA techniques. Generally, in the NOMA scheme, multiple signals are superposed in the power domain and transmitted simultaneously over the same frequency/time channel, and then, successive interference cancellation (SIC) is applied at the receiver side to remove the interference caused by superposition coding \cite{saito2013non}. Consequently, NOMA can provide better performance compared with orthogonal multiple access (OMA) in terms of high spectrum efficiency, massive connectivity, and balancing user fairness. %\cite{chen2016application}. 

Great efforts have been recently carried out to investigate the performance of RIS/STAR-IOS-assisted V2X  communications. In \cite{chen2020resource}, the authors studied the power allocation problem for RIS-aided V2V communication systems based on slowly varying large-scale fading channel information. By considering a stochastic geometry-based analytical framework, the authors in \cite{singh2022visible} analyzed the performance of RIS-aided V2V communication networks in terms of the outage probability (OP), throughput, and delay outage rate (DOR). Given the blockage and vehicle density in practical road conditions, the authors in \cite{wang2020outage} derived a series-form expression of the OP for a RIS-aided vehicular communication network, using the central limit theorem (CLT). By considering the low-latency and high-reliability requirements of V2V links, the authors in \cite{gu2022intelligent} evaluated the ergodic capacity (EC) optimization problem and derived a closed-form expression of optimal phase shifts in vehicle-to-infrastructure (V2I) communications with the help of RIS. The authors in \cite{pala2023design} also optimized the achievable sum-rate for a RIS-aided V2X communication system under full-duplex (FD) mode. %The achievable rate for a RIS-assisted high-mobility vehicular communication system was studied in \cite{huang2021transforming}, proposing a flexible two-step transmission protocol. 
 Aiming to provide reliable and energy-efficient connectivity for autonomous vehicles, the authors in \cite{ozcan2021reconfigurable} optimized the placement problem of RISs over a V2X communication system. Motivated by the potential of RIS in vehicular networks, the secrecy performance metrics for RIS-aided V2V communication under Rayleigh fading channels were also analyzed in \cite{ai2021secure}. By considering a dual RIS-aided V2I communication system, the authors in \cite{shaikh2022performance} derived closed-form expressions of the OP, upper/lower bounds of spectral efficiency (SE), and energy efficiency (EE) under Nakagami-$m$ fading channels. %Only recently, the authors in \cite{aung2023deep} formulated an optimization problem to maximize the achievable rate of users in a STAR-RIS-aided V2X communication system. 
 By comparing the NOMA and OMA schemes, the OP for a STAR-RIS-aided vehicular communication was derived in \cite{guo2023star}. %In addition to aforesaid scientific efforts, several contributions have been also done to analyze the various performance metrics of STAR-RIS-aided NOMA communications in recent years \cite{zhang2022star,wang2022outage,wu2021coverage,ghadi2023analytical,aldababsa2022star,liu2022effective}. 
  However, previous works did not consider RIS and STAR-IOS in vehicular networks simultaneously to analyze the system important performance metrics.   

Motivated by the aforesaid advantages of RIS and STAR-IOS techniques in intelligent practical applications of future wireless networks, e.g., 6G, as well as the inherent potential of the NOMA scheme in providing higher SE over various communication systems, in this paper, we integrate the emerging RIS and STAR-IOS technologies to analyze the performance of vehicular networks, e.g., V2V communication systems, under OMA and NOMA schemes. To the best of the author's knowledge, there has been no previous work that exploits the RIS and   STAR-RIS simultaneously to evaluate V2X communications under composite fading channels. To this end, in particular, we consider a V2V communication, where a transmitter vehicle wants to send an independent message to receiver vehicles with the help of RIS and STAR-IOS which are located close to the transmitter and receivers, respectively. Besides, in order to accurately model the statistical characteristics of fading channel coefficients, we assume that the channels between the RIS and STAR-IOS undergo Fisher-Snedecor $\mathcal{F}$ distribution. Moreover, after determining the marginal distributions of the equivalent channel at the receivers by exploiting the CLT, we derive analytical expressions of important performance metrics in wireless communications, and then, we evaluate the efficiency of the proposed system model in terms of derived metrics. Thus, the main contributions of our work are summarized as follows

\textbullet\, By utilizing the CLT, we first obtain analytical expressions of the cumulative distribution function (CDF) and probability density function (PDF) for the signal-to-noise ratio (SNR) at the receiver vehicles.

\textbullet\, Then, by exploiting the provided CDF and PDF, we derive the closed-form expression of the OP, a tight upper bound of EC, and the analytical expression of EE for both OMA and NOMA schemes under Fisher-Snedecor $\mathcal{F}$ fading channels. 

\textbullet\, Eventually, we evaluate the performance of the considered RIS/STAR-IOS system in terms of the OP, EC, and EE. To do so, we double-check the accuracy of our analytical expressions with Monte-Carlo simulation, where the numerical results confirm that considering RIS/STAR-IOS in V2V communications is quite beneficial for the system performance and that the use of NOMA provides better performance compared to the OMA counterpart.\vspace{-0.5cm}
	\section{System Model}\label{sec-sys}
	\subsection{Channel Model}
	We consider a wireless V2V NOMA communication scenario as shown in Fig. \ref{fig-syst}, wherein a transmitter vehicle $u_\mathrm{s}$ aims to communicate information with  other receiver vehicles $u_\mathrm{w}$, $\mathrm{w}\in\{\mathrm{r,t}\}$, which are located in the reflection and refraction regions. We assume that the direct links between the vehicle transmitter $u_\mathrm{s}$ and the vehicle receivers $u_\mathrm{w}$ are blocked due to obstacles. Hence, we consider a RIS with $N_1$ elements and a STAR-IOS with $N_2$ elements to support this transmission, one each placed near the transmitter vehicle $u_\mathrm{s}$ and the receiver vehicles $u_\mathrm{w}$, respectively. In addition, it is assumed that the distance between the transmitter vehicle and the RIS, $d_\mathrm{sR}$, as well as the distance between the STAR-IOS and receiver vehicles, $d_\mathrm{Sw}$, are small, thereby, the corresponding channels can be properly modeled as deterministic line-of-sight (LoS) channels. However, by considering a large distance between the RIS and the STAR-IOS, $d_\mathrm{RS}$, it is presumed the quasi-static fading channels between the $k$-th element of RIS and the $l$-th element of STAR-IOS follow Fisher-Snedecor $\mathcal{F}$ distribution, i.e., %$h_{k,l}\sim\mathrm{Nak}\left(m,\Omega\right)$, 
	$h_{k,l}\sim\mathcal{F}\left(m_1,m_2\right)$, where $m_1$ and $m_2$ are degrees of freedom which can represent the fading severity parameter and the amount of shadowing of the root-mean-square (rms) signal power, respectively.  %are the shape and spread parameters, respectively. 
	 Therefore, the received signal at $u_\mathrm{w}$ can be respectively expressed as
	\begin{align}
		Y_\mathrm{w}=\sqrt{P}\left(\sum_{k=1}^{N_1}\sum_{l=1}^{N_2}\Phi_kh_{k,l}\Psi_{\mathrm{w},l}\right)\left(\sqrt{p_\mathrm{t}}X_\mathrm{t}+\sqrt{p_\mathrm{r}}X_\mathrm{r}\right)+Z_\mathrm{w}\label{eq-yw},
	\end{align}
in which $P$ denotes the total transmit power, $X_\mathrm{w}$ defines the symbol transmitted to $u_\mathrm{w}$ with unit power (i.e., $\mathbb{E}[|X_\mathrm{w}|^2]$), $p_\mathrm{w}$ is the power allocation factor for $u_\mathrm{w}$, so that $p_\mathrm{t}+p_\mathrm{r}=1$, and $Z_\mathrm{w}$ is the additive white Gaussian noise (AWGN) with zero mean and variance $\sigma_n^2$ at $u_\mathrm{w}$. The term $h_{k,l}=\eta_{k,l}d_{\mathrm{R,S}}^{-\kappa}\mathrm{e}^{-j\zeta_{k,l}}$ denotes the fading channel between the RIS and the STAR-IOS, in which $\kappa>2$, $\eta_{k,l}$, and $\zeta_{k,l}$ are the path-loss exponent, the amplitude of $h_{k,l}$, and the phase of $h_{k,l}$, respectively. Additionally, the term $\Phi_k=\mathrm{e}^{j\phi_k}$ contains $\phi_k$ which indicates the adjustable phase induced by the $k$-th reflecting element of the RIS. By assuming the energy-splitting (ES) protocol for the considered STAR-IOS model, all elements of the STAR-IOS simultaneously operate refraction and reflection modes, while the total radiation energy is split into
two parts, i.e., $\Psi_{\mathrm{w},l}=\beta_{\mathrm{w},l}\mathrm{e}^{j\psi_{\mathrm{w},l}}$, where $\psi_{\mathrm{w},l}$ denotes the adjustable phases induced by the $l$-th element of the STAR-IOS during refraction and reflection, whereas $\beta_{\mathrm{w},l}$ denote the
adjustable refraction/reflection coefficients of the STAR-IOS, with $\beta^2_{\mathrm{r},l}+\beta^2_{\mathrm{t},l}\leq 1$. In order to minimize the complexity of the system, we assume in the sequel that all elements have the same amplitude coefficients, i.e., $\beta_{\mathrm{w},l}=\beta_\mathrm{w}$, $\forall l=1,\dots,N_2$. Moreover, we consider the term $\mathcal{D}$ as the distance-dependent path-loss for LoS which can be defined, for a link distance $d$, at the carrier frequency of $3$ GHz as follows \cite{bjornson2019intelligent}
\begin{align}
\mathcal{D}\left(d\right)\left[dB\right]= -37.5-22\log_{10}\left(d/1 \mathrm{m}\right).\label{eq-loss}
\end{align}
%\begin{align}
%\mathcal{D}\left(d\right)\left[dB\right]=\begin{cases} -37.5-22\log_{10}\left(d/1 \mathrm{m}\right) & \text{if LoS} \\
%	-35.1-36.7\log_{10}\left(d/1 \mathrm{m}\right) &   \text{if NLoS}
%\end{cases}.\label{eq-loss}
%\end{align} 
\subsection{Multiple Access Schemes}
\subsubsection{NOMA}
As per the principles of NOMA, the transmitter vehicle $u_\mathrm{s}$ sends the signals of both receiver vehicles using the same time and frequency resources by superposition coding. Besides, the NOMA vehicle with a better channel condition conducts successive interference cancellation (SIC), while another vehicle decodes its signal directly treating interference as noise. Without loss of generality, since the STAR-IOS is able to adjust the energy allocation coefficients $\beta_{\mathrm{w},l}$ via the ES protocol, we allocate more energy for reflecting links, and thus, the strong receiver vehicle $u_\mathrm{r}$ operates the SIC process. This implies
in turn that the receiver vehicle in the refraction zone $u_\mathrm{t}$ is allocated more power by the transmitter vehicle (i.e., $p_\mathrm{r}<p_\mathrm{t}$). Therefore, the signal-to-interference-plus-noise ratio (SINR) of the SIC process for $u_\mathrm{r}$ can be expressed as
\begin{align}
\gamma_\mathrm{sic}=\frac{\bar{\gamma}p_\mathrm{t}\beta_{\mathrm{r}}^2\mathcal{D}d_{\mathrm{RS}}^{-\kappa}\left|\sum_{k=1}^{N_1}\sum_{l=1}^{N_2}\eta_{k,l}\mathrm{e}^{j\left(\phi_{k}+\psi_{\mathrm{r},l}-\zeta_{k,l}\right)}\right|^2}{\bar{\gamma}p_\mathrm{r}\beta_{\mathrm{r}}^2\mathcal{D}d_{\mathrm{RS}}^{-\kappa}\left|\sum_{k=1}^{N_1}\sum_{l=1}^{N_2}\eta_{k,l}\mathrm{e}^{j\left(\phi_{k}+\psi_{\mathrm{r},l}-\zeta_{k,l}\right)}\right|^2+1},\label{eq-g-sic}
\end{align}
where $\bar{\gamma}=\frac{P}{\sigma^2_n}$ is the transmit SNR. Then, with the aid of SIC, $u_\mathrm{r}$ removes the message of $u_\mathrm{t}$ from its received signal and decodes its required information with the following SNR
\begin{align}
	\gamma_\mathrm{r}^\mathrm{n}=\bar{\gamma}p_\mathrm{r}\beta_{\mathrm{r}}^2\mathcal{D}d_{\mathrm{RS}}^{-\kappa}\left|\sum_{k=1}^{N_1}\sum_{l=1}^{N_2}\eta_{k,l}\mathrm{e}^{j\left(\phi_{k}+\psi_{\mathrm{r},l}-\zeta_{k,l}\right)}\right|^2.\label{eq-gr-noma}
\end{align}
At the same time, $u_\mathrm{t}$ directly decodes its signal by considering the signal of $u_\mathrm{r}$ as interference. Thus, the SINR at $u_\mathrm{t}$ can be expressed as
\begin{align}
\gamma_\mathrm{t}^\mathrm{n}=\frac{\bar{\gamma}p_\mathrm{t}\beta_{\mathrm{t}}^2\mathcal{D}d_{\mathrm{RS}}^{-\kappa}\left|\sum_{k=1}^{N_1}\sum_{l=1}^{N_2}\eta_{k,l}\mathrm{e}^{j\left(\phi_{k}+\psi_{\mathrm{t},l}-\zeta_{k,l}\right)}\right|^2}{\bar{\gamma}p_\mathrm{r}\beta_{\mathrm{t}}^2\mathcal{D}d_{\mathrm{RS}}^{-\kappa}\left|\sum_{k=1}^{N_1}\sum_{l=1}^{N_2}\eta_{k,l}\mathrm{e}^{j\left(\phi_{k}+\psi_{\mathrm{t},l}-\zeta_{k,l}\right)}\right|^2+1}.\label{eq-gt-noma}
\end{align}
\subsubsection{OMA}
On the other hand, the SNR at user $u_\mathrm{w}$ in OMA, such as time division multiple access (TDMA), can be defined as
\begin{align}
	\gamma_\mathrm{w}^\mathrm{o}=\bar{\gamma}\mathcal{D}d_{\mathrm{RS}}^{-\kappa}\beta_\mathrm{w}^2\left|\sum_{k=1}^{N_1}\sum_{l=1}^{N_2}\eta_{k,l}\mathrm{e}^{j\left(\phi_{k}+\psi_{\mathrm{w},l}-\zeta_{k,l}\right)}\right|^2.\label{eq-gw-oma}
\end{align}
\subsubsection{Phase adjustment}
Here, by adjusting the phase at the RIS-to-STAR-IOS channel to cancel the resultant phase, i.e., $\zeta_{k,l}=\phi_k+\psi_{\mathrm{w},l}$, the SNR at receiver vehicles $u_\mathrm{w}$ can be maximized. Therefor, by defining the random variable (RV) $V=\left|\sum_{k=1}^{N_1}\sum_{l=1}^{N_2}\eta_{k,l}\right|$, the maximum SNR at $u_\mathrm{w}$ can be obtained under an ideal phase shift for both NOMA and OMA schemes.
\begin{figure}[!t]
	\centering
	\includegraphics[width=0.9\columnwidth]{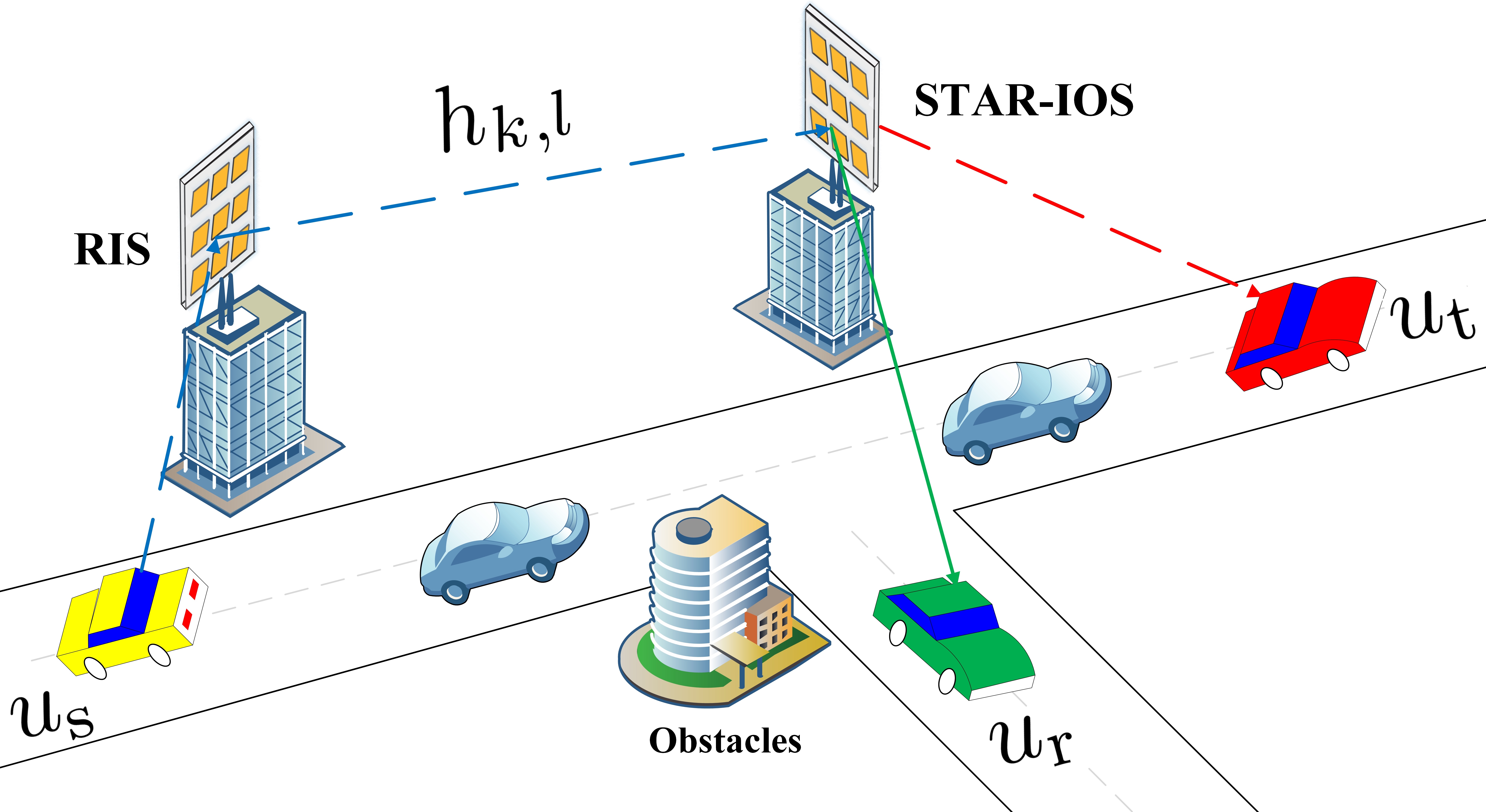}
	\caption{Illustration of RIS/STAR-IOS-aided V2V communications.}\vspace{0cm}
	\label{fig-syst}
\end{figure}
\section{Performance Analysis}
Here, we first introduce the marginal distributions of the SNRs at each users, and then, we derive the closed-form expressions of the OP, EC, and EE under Fisher-Snedecor $\mathcal{F}$ fading channels.  
\subsection{Statistical Characterization}
By exploiting the CLT when $N_1, N_2\gg 1$, the RV $V$ can be accurately approximated by a Gaussian distribution with the following marginal PDF and CDF, respectively
\begin{align}
	f_{V}\left(v\right)=\frac{1}{\sqrt{2\pi\sigma^2_\mathrm{V}}}\exp\left(-\frac{\left(v-\mu_V\right)^2}{2\sigma_\mathrm{V}^2}\right), \quad v>0\label{eq-pdf-v}
\end{align}
\begin{align}
	F_V\left(v\right)=1-\mathrm{Q}\left(\frac{v-\mu_V}{\sigma_V^2}\right), \quad v>0\label{eq-cdf-v}
\end{align}
where $Q(.)$ denotes the $Q$-function.  The terms $\sigma^2_V=\sum_{k=1}^{N_1}\sum_{l=1}^{N_2}\sigma^2_{k,l}$ and $\mu_V=\sum_{k=1}^{N_1}\sum_{l=1}^{N_2}\mu_{k,l}$ are the mean and variance of the RV $V$, respectively. Moreover, $\mu_{k,l}=\frac{m_2}{m_2-1}$ and $\sigma^2_{k,l}=\frac{m_2^2\left(m_1+m_2-1\right)}{m_1\left(m_2-1\right)^2\left(m_2-2\right)}$ are the mean and variance of the RV $\eta_{k,l}$ for all $k=\left\{1,\dots,N_1\right\}$ and $l=\left\{1,\dots,N_2\right\}$, respectively, which undergoes the Fisher-Snedecor $\mathcal{F}$ distribution.

%Moreover, $\mu_{k,l}=\frac{\Gamma\left(m+0.5\right)}{\Gamma\left(m\right)}\sqrt{\frac{\Omega_m}{m}}$ and $\sigma^2_{k,l}=\Omega_m\left(1-\frac{1}{m}\left(\frac{\Gamma\left(m+0.5\right)}{\Gamma\left(m\right)}\right)^2\right)$ are the mean and variance of the RV $\eta_{k,l}$ for all $k=\left\{1,\dots,N_1\right\}$ and $l=\left\{1,\dots,N_2\right\}$, respectively, which undergoes Nakagami-$m$ distribution.
\subsection{Outage Probability}
OP is an appropriate metric to evaluate the performance wireless communication systems, which is defined as the probability that the random SNR $\gamma$ is less than an SNR threshold $\gamma_\mathrm{th}$, i.e., $P_\mathrm{out}=\Pr\left(\gamma\leq\gamma_\mathrm{th}\right)$.
\begin{Proposition}
The OP over the considered dual RIS/STAR-IOS-aided V2V NOMA communications for the receiver vehicles $u_\mathrm{t}$ and $u_\mathrm{r}$ can be respectively given by
\begin{align}
	P_\mathrm{t,out}^\mathrm{n}=1-\mathrm{Q}\left(\frac{\breve{\gamma}_\mathrm{t}-N_1N_2\frac{m_2}{m_2-1}}{N_1N_2\frac{m_2^2\left(m_1+m_2-1\right)}{m_1\left(m_2-1\right)^2\left(m_2-2\right)}}\right),\label{eq-outt-noma}
\end{align}
%\begin{align}
%P_\mathrm{t,out}^\mathrm{n}=1-\mathrm{Q}\left(\frac{\breve{\gamma}_\mathrm{t}-N_1N_2\frac{\Gamma\left(m+0.5\right)}{\Gamma\left(m\right)}\sqrt{\frac{\Omega_m}{m}}}{N_1N_2\Omega_m\left(1-\frac{1}{m}\left(\frac{\Gamma\left(m+0.5\right)}{\Gamma\left(m\right)}\right)^2\right)}\right),\label{eq-outt-noma}
%\end{align}
\begin{align}\nonumber
	P_{\mathrm{r,out}}^\mathrm{n}=&\left[1-\mathrm{Q}\left(\frac{\hat{\gamma}_\mathrm{sic}-N_1N_2\frac{m_2}{m_2-1}}{N_1N_2\frac{m_2^2\left(m_1+m_2-1\right)}{m_1\left(m_2-1\right)^2\left(m_2-2\right)}}\right)\right]\\
	&\times\left[1-\mathrm{Q}\left(\frac{\check{\gamma}_\mathrm{r}-N_1N_2\frac{m_2}{m_2-1}}{N_1N_2\frac{m_2^2\left(m_1+m_2-1\right)}{m_1\left(m_2-1\right)^2\left(m_2-2\right)}}\right)\right],\label{eq-outr-noma}
\end{align}
%\begin{align}\nonumber
%&P_{\mathrm{r,out}}^\mathrm{n}=\left[1-\mathrm{Q}\left(\frac{\hat{\gamma}_\mathrm{sic}-N_1N_2\frac{\Gamma\left(m+0.5\right)}{\Gamma\left(m\right)}\sqrt{\frac{\Omega_m}{m}}}{N_1N_2\Omega_m\left(1-\frac{1}{m}\left(\frac{\Gamma\left(m+0.5\right)}{\Gamma\left(m\right)}\right)^2\right)}\right)\right]\\
%&\times\left[1-\mathrm{Q}\left(\frac{\check{\gamma}_\mathrm{r}-N_1N_2\frac{\Gamma\left(m+0.5\right)}{\Gamma\left(m\right)}\sqrt{\frac{\Omega_m}{m}}}{N_1N_2\Omega_m\left(1-\frac{1}{m}\left(\frac{\Gamma\left(m+0.5\right)}{\Gamma\left(m\right)}\right)^2\right)}\right)\right],\label{eq-outr-noma}
%\end{align}
where $\breve{\gamma}_\mathrm{t}=\sqrt{\frac{\bar{\gamma}_\mathrm{t}^\mathrm{n}}{\bar{\gamma}\beta_{\mathrm{t}}^2\mathcal{D}d_{\mathrm{RS}}^{-\kappa}\left(p_\mathrm{t}-p_\mathrm{r}\bar{\gamma}_\mathrm{t}^\mathrm{n}\right)}}$, $\hat{\gamma}_\mathrm{sic}=\sqrt{\frac{\bar{\gamma}_\mathrm{sic}}{\bar{\gamma}\beta_{\mathrm{r}}^2\mathcal{D}d_{\mathrm{RS}}^{-\kappa}\left(p_\mathrm{r}-p_\mathrm{t}\bar{\gamma}_\mathrm{sic}\right)}}$, and $\check{\gamma}_\mathrm{r}=\sqrt{\frac{\bar{\gamma}_\mathrm{r}^\mathrm{n}}{\beta_{\mathrm{r}}^2\bar{\gamma}p_\mathrm{r}\mathcal{D}d_{\mathrm{RS}}^{-\kappa}}}$. Besides, $\breve{\gamma}_\mathrm{t}$, $\hat{\gamma}_\mathrm{sic}$, and $\check{\gamma}_\mathrm{r}$ are the SNR thresholds of $\gamma_\mathrm{t}$, $\gamma_\mathrm{sic}$, and $\gamma_\mathrm{r}$, respectively. 
\end{Proposition}
\begin{proof}
Regarding the principle of NOMA, the outage occurs at $u_\mathrm{r}$ when it cannot decode the signal of $u_\mathrm{t}$ or its own signal,
or both. Hence, the OP of $u_\mathrm{r}$ is defined as
\begin{align}
	P_\mathrm{r,out}^\mathrm{n}=\Pr\left(\gamma_\mathrm{sic}\leq \bar{\gamma}_\mathrm{sic},\gamma_\mathrm{r}^\mathrm{n}\leq\bar{\gamma}_\mathrm{r}^\mathrm{n}\right),\label{eqp-outr-def}
	\end{align}
where $\bar{\gamma}_\mathrm{sic}$ and $\bar{\gamma}_\mathrm{r}^\mathrm{n}$ are the corresponding SINR/SNR thresholds. In this case, it should be noted that the OP highly depends on the relation between power allocation factors $p_\mathrm{r}$ and $p_\mathrm{t}$. Hence, when $p_\mathrm{t}\leq{\bar{\gamma}_\mathrm{sic}}p_\mathrm{r}$, we will have $P_\mathrm{r,out}^\mathrm{n}=1$; Otherwise, when $p_\mathrm{t}>{\bar{\gamma}_\mathrm{sic}}p_\mathrm{r}$, by considering the ideal phase shift to reach the maximum received signal, and then, inserting \eqref{eq-g-sic} and \eqref{eq-gr-noma} into \eqref{eqp-outr-def}, the OP of $u_\mathrm{r}$ can be given by
\begin{align}\nonumber
	&P_\mathrm{r,out}^\mathrm{n}\overset{(a)}{=}\Pr\left(\frac{\bar{\gamma}p_\mathrm{r}\mathcal{D}d_{\mathrm{RS}}^{-\kappa}\beta_{\mathrm{r}}^2V^2}{\bar{\gamma}p_\mathrm{t}\mathcal{D}d_{\mathrm{RS}}^{-\kappa}\beta_{\mathrm{r}}^2V^2+1}\leq \bar{\gamma}_\mathrm{sic}\right)\\
	&\quad\quad\quad\times\Pr\left(\bar{\gamma}p_\mathrm{r}\mathcal{D}d_{\mathrm{RS}}^{-\kappa}\beta_{\mathrm{r}}^2V^2\leq \bar{\gamma}_\mathrm{r}^\mathrm{n}\right)\\\nonumber
	&=\Pr\left(V\leq\sqrt{\frac{\bar{\gamma}_\mathrm{sic}}{\bar{\gamma}\beta_{\mathrm{r}}^2\mathcal{D}d_{\mathrm{RS}}^{-\kappa}\left(p_\mathrm{r}-p_\mathrm{t}\bar{\gamma}_\mathrm{sic}\right)}}\right)\\
	&\quad\quad\quad\times\Pr\left(V\leq\sqrt{\frac{\bar{\gamma}_\mathrm{r}^\mathrm{n}}{\bar{\gamma}p_\mathrm{r}\mathcal{D}d_{\mathrm{RS}}^{-\kappa}\beta_{\mathrm{r}}^2}}\right)\\
	&=F_V\left(\hat{\gamma}_\mathrm{sic}\right)F_V\left(\check{\gamma}_\mathrm{r}\right)\\
	&=\left[1-\mathrm{Q}\left(\frac{\hat{\gamma}_\mathrm{sic}-\mu_V}{\sigma_V^2}\right)\right]\left[1-\mathrm{Q}\left(\frac{\check{\gamma}_\mathrm{r}-\mu_V}{\sigma_V^2}\right)\right],
\end{align}
where $(a)$ is obtained from the independence of the events. Next, by substituting the values of $\mu_V$ and $\sigma_V^2$ that previously obtained, the proof of $P_\mathrm{r,out}^\mathrm{n}$ is accomplished. 

As for the proof of $P_\mathrm{t,out}^\mathrm{n}$, the outage occurs at $u_\mathrm{t}$ when its received SINR $\gamma_\mathrm{t}^\mathrm{n}$ falls below a SINR threshold $\bar{\gamma}_\mathrm{t}^\mathrm{n}$, i.e., 
\begin{align}
	P_\mathrm{t,out}^\mathrm{n}=\Pr\left(\gamma_\mathrm{t}^\mathrm{n}\leq\bar{\gamma}_\mathrm{t}^\mathrm{n}\right).\label{eqp-outt-def}
	\end{align}
Now, following the similar procedures as in the analysis of $u_\mathrm{r}$, under the case of $p_\mathrm{t}\leq{\bar{\gamma}_\mathrm{t}}^\mathrm{n}p_\mathrm{r}$, we will have $P_\mathrm{t,out}^\mathrm{n}=1$. However, when $p_\mathrm{t}>{\bar{\gamma}_\mathrm{t}}^\mathrm{n}p_\mathrm{r}$, by inserting \eqref{eq-gt-noma} into \eqref{eqp-outt-def}, the OP of $u_\mathrm{t}$ can be given by
\begin{align}
	P_\mathrm{t,out}^\mathrm{n}&=\Pr\left(\frac{\bar{\gamma}p_\mathrm{t}\mathcal{D}d_{\mathrm{RS}}^{-\kappa}\beta_{\mathrm{t}}^2V^2}{\bar{\gamma}p_\mathrm{r}\mathcal{D}d_{\mathrm{RS}}^{-\kappa}\beta_{\mathrm{t}}^2V^2+1}\leq \bar{\gamma}_\mathrm{t}^\mathrm{n}\right)\\
	&=\Pr\left(V\leq\sqrt{\frac{\bar{\gamma}_\mathrm{t}^\mathrm{n}}{\bar{\gamma}\beta_{\mathrm{t}}^2\mathcal{D}d_{\mathrm{RS}}^{-\kappa}\left(p_\mathrm{t}-p_\mathrm{r}\bar{\gamma}_\mathrm{t}^\mathrm{n}\right)}}\right)\\
	&=F_V\left(\breve{\gamma}_\mathrm{t}\right)=1-\mathrm{Q}\left(\frac{\breve{\gamma}_\mathrm{w}-\mu_V}{\sigma_V^2}\right),
\end{align}
then, by substituting the values of $\mu_V$ and $\sigma_V^2$ that previously obtained, the proof of $P_\mathrm{t,out}^\mathrm{n}$ is completed. 
\end{proof}
\begin{Proposition}
	The OP over the considered dual RIS/STAR-IOS-aided V2V OMA communications for the receiver vehicles $u_\mathrm{w}$ can be given by
		\begin{align}
		P_\mathrm{w,out}^\mathrm{o}=1-\mathrm{Q}\left(\frac{\tilde{\gamma}_\mathrm{w}-N_1N_2\frac{m_2}{m_2-1}}{N_1N_2\frac{m_2^2\left(m_1+m_2-1\right)}{m_1\left(m_2-1\right)^2\left(m_2-2\right)}}\right),
	\end{align}
%	\begin{align}
%		P_\mathrm{w,out}^\mathrm{o}=1-\mathrm{Q}\left(\frac{\tilde{\gamma}_\mathrm{w}-N_1N_2\frac{\Gamma\left(m+0.5\right)}{\Gamma\left(m\right)}\sqrt{\frac{\Omega_m}{m}}}{N_1N_2\Omega_m\left(1-\frac{1}{m}\left(\frac{\Gamma\left(m+0.5\right)}{\Gamma\left(m\right)}\right)^2\right)}\right),
%	\end{align}
where $\tilde{\gamma}_\mathrm{w}=\sqrt{\frac{\bar{\gamma}_\mathrm{w}^\mathrm{o}}{\bar{\gamma}\mathcal{D}d_{\mathrm{RS}}^{-\kappa}\beta_{\mathrm{w}}^2}}$ and $\bar{\gamma}_\mathrm{w}^\mathrm{o}$ denotes the SNR threshold of $u_\mathrm{w}$.
	\end{Proposition}
\begin{proof}
As for the OMA scheme, the OP of receiver vehicles $u_\mathrm{w}$ can be defined as the probability that $\gamma_\mathrm{w}^\mathrm{o}$ less than the SNR threshold $\bar{\gamma}_\mathrm{w}^\mathrm{o}$, i.e., 
\begin{align}
	P_\mathrm{w,out}^\mathrm{o}=\Pr\left(\gamma_\mathrm{w}\leq \bar{\gamma}_\mathrm{w}^\mathrm{o}\right).\label{eqp-outw-def}
	\end{align}
Next, by plugging \eqref{eq-gw-oma} into \eqref{eqp-outw-def} and considering the ideal phase shift, we have
\begin{align}
	P_\mathrm{w,out}^\mathrm{o}&=\Pr\left(V\leq\sqrt{\frac{\bar{\gamma}_\mathrm{w}^\mathrm{o}}{\bar{\gamma}\mathcal{D}d_{\mathrm{RS}}^{-\kappa}\beta_{\mathrm{w}}^2}}\right)\\
	&=F_V\left(\tilde{\gamma}_\mathrm{w}\right)=1-\mathrm{Q}\left(\frac{\tilde{\gamma}_\mathrm{w}-\mu_V}{\sigma_V^2}\right).
\end{align}
Now, by substituting the values of $\mu_V$ and $\sigma_V^2$ that previously obtained, the proof of $P_\mathrm{w,out}^\mathrm{o}$ is completed. 
\end{proof}
\subsection{Ergodic Capacity}
The EC for the considered system model under NOMA/OMA schemes with the instantaneous SNR $\gamma_\mathrm{w}^\mathrm{z}$, $\mathrm{z}\in\{\mathrm{n,o}\}$,  can be defined as
\begin{align}
\hspace{-0.4cm}\bar{C}_\mathrm{w}^\mathrm{z}&=\mathbb{E}\left[\log_2\left(1+\gamma_\mathrm{w}^\mathrm{z}\right)\right]=\int_0^\infty \log_2\left(1+\gamma_\mathrm{w}^\mathrm{z}\right)f_V\left(v\right)dv.\label{eq-c-def}
\end{align}
In general, it is mathematically intractable to derive the closed-form expression of $\bar{C}_\mathrm{w}^\mathrm{z}$. Hence, by exploiting the Jensen's inequality, we derive tight upper bound $\bar{C}_\mathrm{w,U}^\mathrm{z}$ %and lower bound $\bar{C}_\mathrm{w,L}^\mathrm{z}$ 
 for the EC, i.e., %$\bar{C}_\mathrm{w,L}^\mathrm{z}\leq 
 $\bar{C}_\mathrm{w}^\mathrm{z}\leq \bar{C}_\mathrm{w,U}^\mathrm{z}$. %$\bar{C}_\mathrm{w}^\mathrm{z}\leq \bar{C}_\mathrm{w,up}^\mathrm{z}=\log_2\left(1+\mathbb{E}\left[\gamma_\mathrm{w}^\mathrm{z}\right]\right)$. 
\begin{Proposition}
The upper bound EC over the considered dual RIS/STAR-IOS-aided V2V NOMA communications for the receiver vehicles $u_\mathrm{r}$ and $u_\mathrm{t}$ can be respectively given by \eqref{eq-cr-noma} and \eqref{eq-ct-noma}.
\begin{figure*}[t]
	\normalsize
	%\hrulefill
	\setcounter{equation}{24}
	\normalsize\begin{align}
		\bar{C}_\mathrm{r,U}^\mathrm{n}=\log_2\left(1+\bar{\gamma}p_\mathrm{r}\mathcal{D}d_{\mathrm{RS}}^{-\kappa}\beta_{\mathrm{r}}^2N_1N_2\left[\frac{m_2}{m_2-1}+\frac{m_2^2\left(m_1+m_2-1\right)}{m_1\left(m_2-1\right)^2\left(m_2-2\right)}\right]\right).\label{eq-cr-noma}
	\end{align}
	\hrulefill\vspace{0cm}
\end{figure*}
%\begin{figure*}[t]
%	\normalsize
%	%\hrulefill
%	\setcounter{equation}{26}
%	\normalsize\begin{align}
%		\bar{C}_\mathrm{r,U}^\mathrm{n}=\log_2\left(1+\bar{\gamma}p_\mathrm{r}\mathcal{D}\beta_{\mathrm{r}}^2N_1N_2\Omega_m\left[1+\left(\frac{N_1N_2-1}{m}\right)\left(\frac{\Gamma\left(m+0.5\right)}{\Gamma\left(m\right)}\right)^2\right]\right).\label{eq-cr-noma}
%	\end{align}
%	\hrulefill\vspace{0cm}
%\end{figure*}
\begin{figure*}[t]
	\normalsize
	%\hrulefill
	\setcounter{equation}{25}
	\normalsize
	\begin{align}
		\bar{C}_\mathrm{t,U}^\mathrm{n}=\log_2\left(1+\frac{\bar{\gamma}p_\mathrm{t}\mathcal{D}d_{\mathrm{RS}}^{-\kappa}\beta_{\mathrm{t}}^2N_1N_2\left[\frac{m_2}{m_2-1}+\frac{m_2^2\left(m_1+m_2-1\right)}{m_1\left(m_2-1\right)^2\left(m_2-2\right)}\right]}{\bar{\gamma}p_\mathrm{r}\mathcal{D}d_{\mathrm{RS}}^{-\kappa}\beta_{\mathrm{t}}^2N_1N_2\left[\frac{m_2}{m_2-1}+\frac{m_2^2\left(m_1+m_2-1\right)}{m_1\left(m_2-1\right)^2\left(m_2-2\right)}\right]+1}\right).\label{eq-ct-noma}
	\end{align}
	\hrulefill\vspace{0cm}
\end{figure*}
%\begin{figure*}[t]
%	\normalsize
%	%\hrulefill
%	\setcounter{equation}{27}
%	\normalsize
%	\begin{align}
%		\bar{C}_\mathrm{t,U}^\mathrm{n}=\log_2\left(1+\frac{\bar{\gamma}p_\mathrm{t}\mathcal{D}\beta_{\mathrm{t}}^2N_1N_2\Omega_m\left[1+\left(\frac{N_1N_2-1}{m}\right)\left(\frac{\Gamma\left(m+0.5\right)}{\Gamma\left(m\right)}\right)^2\right]}{\bar{\gamma}p_\mathrm{r}\mathcal{D}\beta_{\mathrm{t}}^2N_1N_2\Omega_m\left[1+\left(\frac{N_1N_2-1}{m}\right)\left(\frac{\Gamma\left(m+0.5\right)}{\Gamma\left(m\right)}\right)^2\right]+1}\right).\label{eq-ct-noma}
%	\end{align}
%	\hrulefill\vspace{0cm}
%\end{figure*}
\end{Proposition}
\begin{proof}
By inserting \eqref{eq-gr-noma} and \eqref{eq-pdf-v} into \eqref{eq-c-def} and applying the Jensen's inequality, the upper bound  EC of $u_\mathrm{r}$ under NOMA scheme can be defined as
\begin{align}
\bar{C}_\mathrm{r,U}^\mathrm{n}=\log_2\left(1+\bar{\gamma}p_\mathrm{r}\mathcal{D}d_{\mathrm{RS}}^{-\kappa}\beta_{\mathrm{r}}^2\mathbb{E}\left[V^2\right]\right),\label{eqp-c-e2}
\end{align}
in which $\mathbb{E}\left[V^2\right]$ can be determined as 
\begin{align}
&\mathbb{E}\left[V^2\right]=\text{Var}\left[V\right]+\mathbb{E}^2\left[V\right]=\sigma_V^2+\mu_V^2.\\
%&=N_1N_2\Omega_m\left[1+\left(\frac{N_1N_2-1}{m}\right)\left(\frac{\Gamma\left(m+0.5\right)}{\Gamma\left(m\right)}\right)^2\right]. \label{eqp-e2}
&=N_1N_2\left[\frac{N_1N_2m_2^2}{\left(m_2-1\right)^2}+\frac{m_2^2\left(m_1+m_2-1\right)}{m_1\left(m_2-1\right)^2\left(m_2-2\right)}\right].\label{eqp-e2}
\end{align}
Now, by substituting the value of $\mathbb{E}\left[V^2\right]$ into \eqref{eqp-c-e2}, $\bar{C}_\mathrm{r,U}^\mathrm{o}$ is derived as \eqref{eq-cr-noma}.

Similarly, by plugging \eqref{eq-gt-noma} and \eqref{eq-pdf-v} into \eqref{eq-c-def} and considering the Jensen's inequality, the upper bound EC of $u_\mathrm{t}$ under NOMA case can be expressed as 
\begin{align}
	\bar{C}_\mathrm{t,U}^\mathrm{n}=\log_2\left(1+\frac{\bar{\gamma}p_\mathrm{t}\mathcal{D}d_{\mathrm{RS}}^{-\kappa}\beta_{\mathrm{t}}^2\mathbb{E}\left[V^2\right]}{\bar{\gamma}p_\mathrm{r}\mathcal{D}d_{\mathrm{RS}}^{-\kappa}\beta_{\mathrm{t}}^2\mathbb{E}\left[V^2\right]+1}\right),\label{eqp-ct-e2}
\end{align}
in which by inserting the values of  $\mathbb{E}\left[V^2\right]$ from \eqref{eqp-e2} into \eqref{eqp-ct-e2}, the proof is completed. 
\end{proof}
\begin{Proposition}
The upper bound EC over the considered dual RIS/STAR-IOS-aided V2V OMA communications for the receiver vehicles $u_\mathrm{w}$ can be given by \eqref{eq-cw-oma}.
\begin{figure*}[t]
	\normalsize
	%\hrulefill
	\setcounter{equation}{30}
	\normalsize
	\begin{align}
		\bar{C}_\mathrm{w,U}^\mathrm{o}=\log_2\left(1+\bar{\gamma}\mathcal{D}d_{\mathrm{RS}}^{-\kappa}\beta_\mathrm{w}^2N_1N_2\left[\frac{m_2}{m_2-1}+\frac{m_2^2\left(m_1+m_2-1\right)}{m_1\left(m_2-1\right)^2\left(m_2-2\right)}\right]\right).\label{eq-cw-oma}
	\end{align}
	\hrulefill\vspace{0cm}
\end{figure*}
%\begin{figure*}[t]
%	\normalsize
%	%\hrulefill
%	\setcounter{equation}{32}
%	\normalsize
%	\begin{align}
%		\bar{C}_\mathrm{w,U}^\mathrm{o}=\log_2\left(1+\bar{\gamma}\mathcal{D}\beta_\mathrm{w}^2N_1N_2\Omega_m\left[1+\left(\frac{N_1N_2-1}{m}\right)\left(\frac{\Gamma\left(m+0.5\right)}{\Gamma\left(m\right)}\right)^2\right]\right)\label{eq-cw-oma}
%	\end{align}
%	\hrulefill\vspace{0cm}
%\end{figure*}
\end{Proposition}
\begin{proof}
By inserting \eqref{eq-gw-oma} and \eqref{eq-pdf-v} into \eqref{eq-c-def} and applying the Jensen's inequality, the upper bound EC of $u_\mathrm{w}$ under OMA scheme can be defined as
\begin{align}
	\bar{C}_\mathrm{w,U}^\mathrm{o}=\log_2\left(1+\bar{\gamma}\mathcal{D}d_{\mathrm{RS}}^{-\kappa}\beta_\mathrm{w}^2\mathbb{E}\left[V^2\right]\right).\label{eqp-cw-e2}
\end{align}
Now, by substituting \eqref{eqp-e2} into \eqref{eqp-cw-e2} the proof is accomplished. 
\end{proof}\vspace{-1cm}
\subsection{Energy Efficiency}
EE is a vital performance metric due to limited resources in wireless communication networks. In this regard, the EE for the considered dual RIS/STAR-IOS-aided NOMA/OMA communications is defined as the ratio of the sum of EC at $u_\mathrm{w}$ to the corresponding total power consumption $P_\mathrm{tot}$, i.e.,
\begin{align}
\mathcal{E}^\mathrm{z}&=\frac{\bar{C}_\mathrm{t}^\mathrm{z}+\bar{C}_\mathrm{r}^\mathrm{z}}{P_\mathrm{tot}}=\frac{\bar{C}_\mathrm{t}^\mathrm{z}+\bar{C}_\mathrm{r}^\mathrm{z}}{P/\alpha+N_1P_\mathrm{R}+N_2P_\mathrm{S}+P_\mathrm{t}+P_\mathrm{r}},\label{eq-ee-z}
\end{align}
where $P/\alpha$ is the dynamic power consumption at transmitter vehicle $u_\mathrm{s}$ in which $\alpha$ indicates the drain efficiency of high-power amplifier (HPA). The terms $P_\mathrm{R}$ and $P_\mathrm{S}$ denote the power consumed by each element of the RIS and STAR-IOS, respectively. Besides, $P_\mathrm{t}$ and $P_\mathrm{r}$ are the circuit power consumption at receiver vehicles $u_\mathrm{t}$ and $u_\mathrm{r}$, respectively.

Here, by applying the obtained ergodic capacity of NOMA and OMA schemes to \eqref{eq-ee-z}, the corresponding EE can be determined.\vspace{-0.5cm}
\begin{figure*}\hspace{-1cm}
	\subfigure[]{%
		\includegraphics[width=0.28\textwidth]{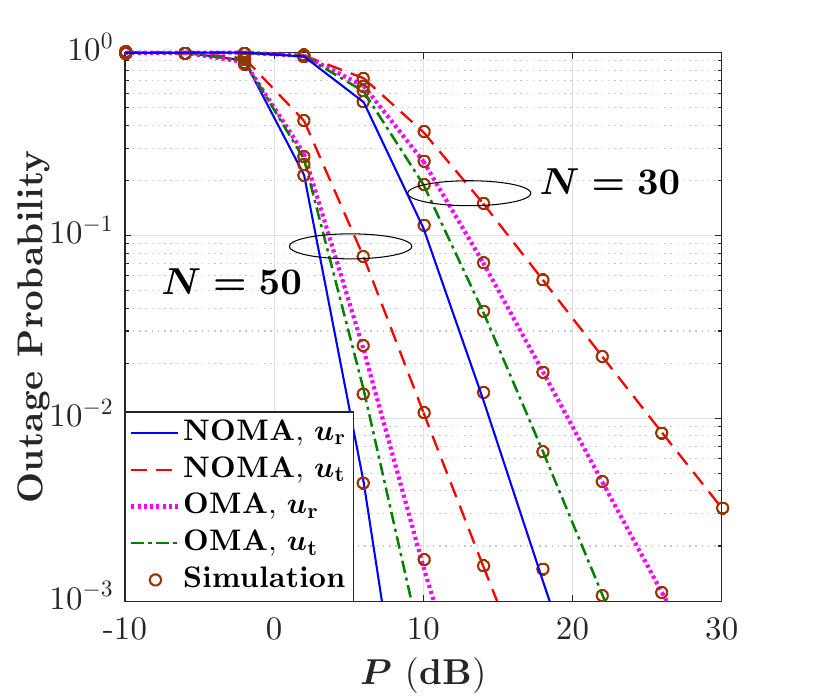}\label{fig-P-p}%
	}\hspace{-0.4cm}%or more
	\subfigure[]{%
		\includegraphics[width=0.28\textwidth]{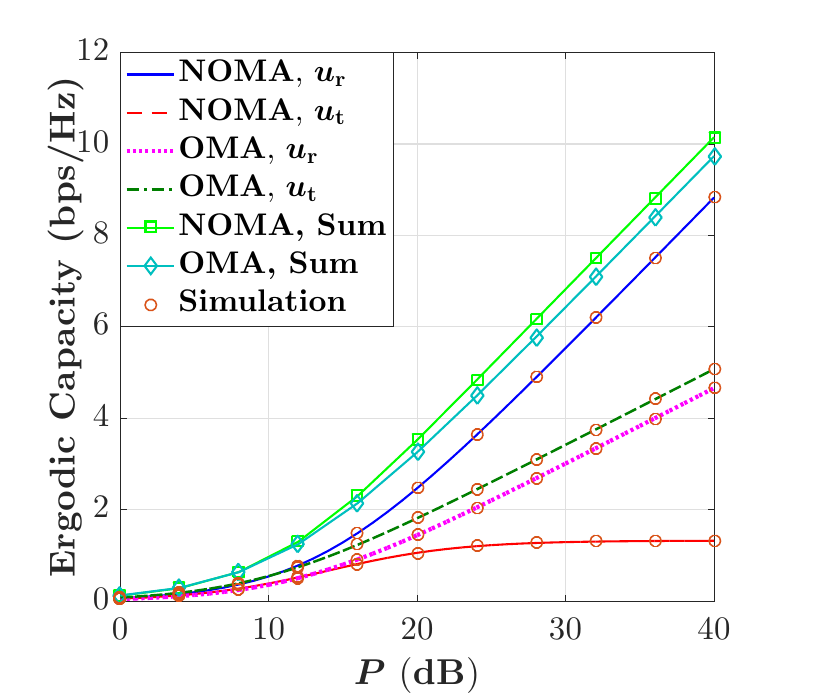}\label{fig-EC-p}%
	}\hspace{-0.53cm}%or more
	\subfigure[]{%
		\includegraphics[width=0.28\textwidth]{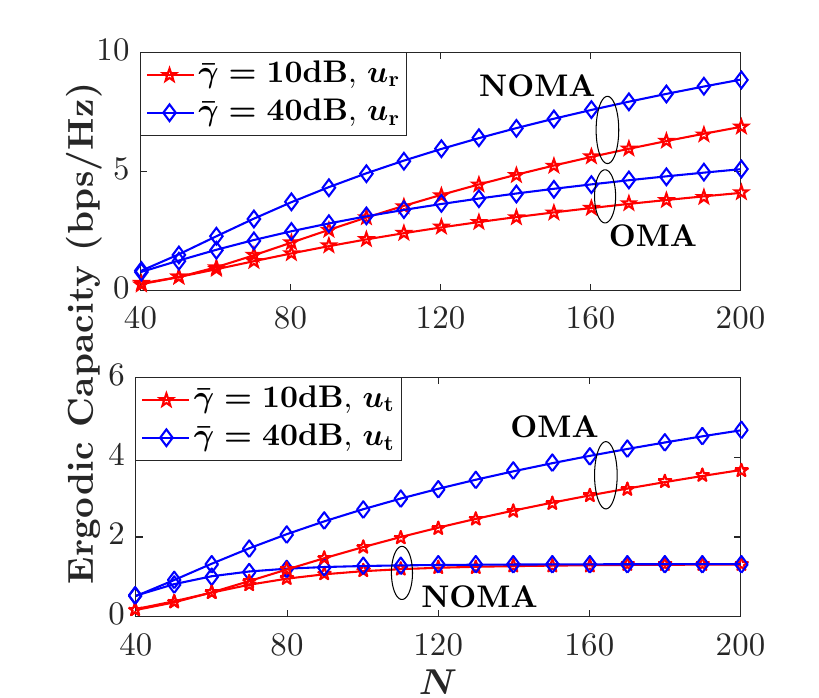}\label{fig-EC-n}%
	}\hspace{-0.46cm}
	\subfigure[]{%
	\includegraphics[width=0.28\textwidth]{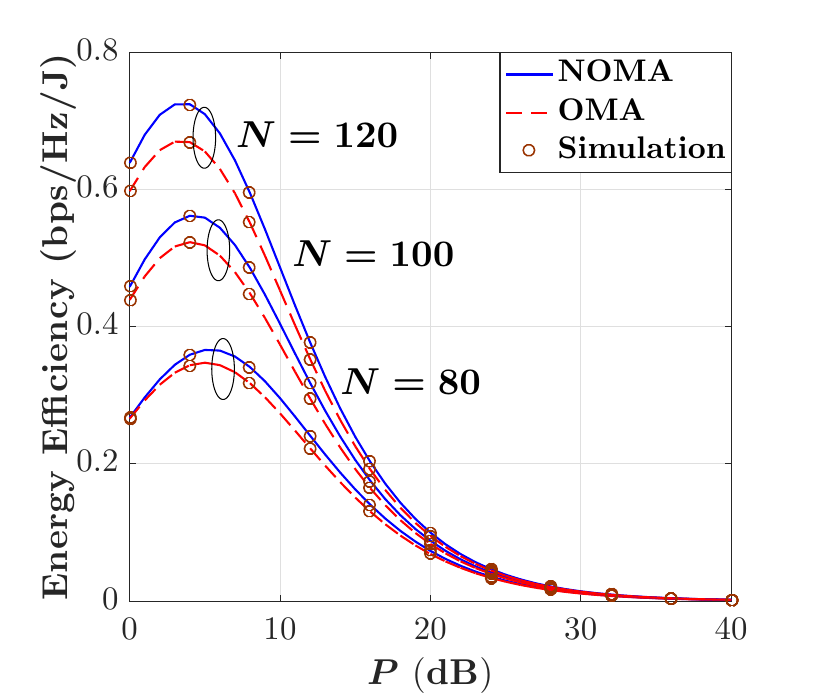}\label{fig-EE-p}%
}
\hspace{-1.5cm}
	\caption{(a) OP versus transmit power $P$ for different values of RIS/STAR-IOS elements $N$; (b) EC versus transmit power $P$ for a given value of RIS/STAR-IOS elements $N=50$; (c) EC versus the number of RIS/STAR-IOS elements $N$ for selected values of average SNR $\bar{\gamma}$; (d) EE versus the transmit power $P$ for selected values of RIS/STAR-IOS elements $N$.}\label{fig-outg}\vspace{-0.6cm}
\end{figure*}
\section{Numerical Results}\label{sec-num}
In this section, we present the numerical results to validate the theoretical expressions previously derived, which are double-checked in all instances with Monte Carlo (MC) simulations. We set the simulation parameters as $P=30$dBm, $\beta_{\mathrm{r}}=0.8$, $\beta_{\mathrm{t}}=0.6$, $p_\mathrm{r}=0.4$, $p_\mathrm{t}=0.6$, $d_\mathrm{s,R}=d_\mathrm{S,t}=d_\mathrm{S,r}=10$m, $\kappa=4$, $d_\mathrm{RS}=100$m, $\sigma^2=-90$dBm, $\alpha=1.2$, $P_\mathrm{R}=P_\mathrm{S}=P_\mathrm{t}=P_\mathrm{r}=10$dBm, and $(m_1,m_2)=(1,3)$. In addition, for the simulation purpose, we assume that $N_1=N_2=N$.

Fig. \ref{fig-P-p} presents the performance of the OP in terms of the transmit power $P$ for different numbers of RIS/STAR-IOS elements $N$ in the considered OMA/NOMA scenario under Fisher-Snedecor $\mathcal{F}$ fading channels. It can be seen that as the transmit power increases, the OP decreases for both receiver vehicles $u_\mathrm{w}$ under NOMA and OMA scenarios, which is reasonable since the channel conditions become better. We can also see that the NOMA scheme can provide better OP performance for receiver vehicle $u_\mathrm{r}$ compared with the OMA case since $u_\mathrm{r}$ exploits the SIC in the NOMA scenario. However, it can be observed that the OP behavior for the receiver vehicle $u_\mathrm{t}$ improves under OMA scheme compared with the NOMA case. Moreover, it is clearly seen that increasing the number of RIS/STAR-IOS elements $N$ can significantly improve the OP performance of $u_\mathrm{w}$ under both OMA and NOMA case since such increment can remarkably enhance the corresponding spatial diversity.

%\begin{figure}[!t]
%	\centering
%	\includegraphics[width=0.9\columnwidth]{P-p1.eps}
%	\caption{OP versus transmit power $P$ for different values of RIS/STAR-IOS elements $N$.}\vspace{0cm}
%	\label{fig-P-p}
%\end{figure}
The behavior of EC in terms of the transmit power $P$ for a selected number of RIS/STAR-IOS elements $N=50$ in the considered V2V OMA/NOMA communication system is illustrated in Fig. \ref{fig-EC-p}. In the NOMA case, we can see that the EC performance for the receiver vehicle $u_\mathrm{r}$ is better than that of receiver vehicle $u_\mathrm{t}$ across the same transmit power $P$, where this superiority becomes more noticeable in the high transmit power area. This behavior indicates that the EC for $u_\mathrm{t}$ increases linearly with $P$, while for $u_\mathrm{t}$ it converges to a constant as $P$ increases. The main reason for these behaviors is that in the NOMA scheme $\mathrm{u}_r$ benefits from SIC, thereby, it has a larger SINR value compared with $u_\mathrm{t}$. It can be also seen that under the OMA scenario, the EC value increases continuously for both receiver vehicles $u_\mathrm{r}$ and $u_\mathrm{t}$  as the $P$ grows. Moreover, we can observe that the EC performance improves only for $u_\mathrm{r}$ while it deteriorates for $u_\mathrm{t}$ compared with the NOMA case. However, by comparing OMA and NOMA schemes, it can be seen that the higher system sum EC is provided by the NOMA case which means that NOMA is quite beneficial for the considered RIS/STAR-IOS communication.  
%\begin{figure}[!t]
%	\centering
%	\includegraphics[width=0.9\columnwidth]{EC-p.eps}
%	\caption{EC versus transmit power $P$ for a given value of RIS/STAR-IOS elements $N=50$.}\vspace{0cm}
%	\label{fig-EC-p}
%\end{figure}
%\begin{figure}[!t]
%	\centering
%	\includegraphics[width=0.9\columnwidth]{EC-n.eps}
%	\caption{EC versus the number of RIS/STAR-IOS elements $N$ for selected values of average SNR $\bar{\gamma}$.}\vspace{0cm}
%	\label{fig-EC-n}
%\end{figure}
%\begin{figure}[!t]
%	\centering
%	\includegraphics[width=0.9\columnwidth]{EE-p.eps}
%	\caption{EE versus the transmit power $P$ for selected values of RIS/STAR-IOS elements $N$.}\vspace{0cm}
%	\label{fig-EE-p}
%\end{figure}
Fig. \ref{fig-EC-n} presents more insights into the effect of the number of RIS/STAR-IOS elements on the EC performance for the considered OMA/NOMA communication system. For both OMA and NOMA schemes, we can observe that the EC performance for vehicles $u_\mathrm{w}$ improves as the number of RIS/STAR-RIS elements increases. As expected, increasing the number of RIS/STAR-IOS leads to an improvement in spatial diversity, thus, enhancing the receiver vehicles' performance. It can be also seen that increasing the average SNR $\bar{\gamma}$ provides a higher EC for the vehicles $u_\mathrm{w}$ under both OMA and NOMA cases. Additionally, under NOMA scheme, we can observe that $\bar{\gamma}$ has a greater impact on the EC performance of vehicle receiver $u_\mathrm{t}$ when $N$ is small, however, when $N$ increases to a certain value, $\bar{\gamma}$ does not affect the EC performance of $u_\mathrm{t}$ anymore. The reason for this behavior is that when $N$ grows, the SINR of $u_\mathrm{t}$ reaches to a constant value, i.e., $N\rightarrow\infty, \gamma_\mathrm{t}^\mathrm{n}\rightarrow p_\mathrm{t}/p_\mathrm{r}$.

Fig. \ref{fig-EE-p} shows the performance of EE in terms of the transmit power $P$ for selected values of the number of RIS/STAR-IOS elements $N$ under OMA and NOMA scenarios. We can see that by increasing the transmit power $P$, EE initially increases and reaches its highest point at approximately the value of $P=5$dB, and then drops. This is because the EC is more dominant than power consumption before reaching the extreme point, however, this effect will reverse as $P$ grows after reaching the extreme point. It is also observed that a higher value of EE can be provided as the number of RIS/STAR-IOS elements increases which is reasonable since the EC improves under such conditions. Furthermore, it can be obtained that the NOMA case provides a better performance in terms of EE for the considered RIS/STAR-IOS communication system compared with the OMA scheme.
\section{Conclusion}\label{sec-con}
In this paper, we integrated the RIS/STAR-IOS into vehicular networks to evaluate the performance of V2V communications. To this end, by considering the OMA and NOMA scenarios, we assumed that a transmitter vehicle wants to send an independent message with the help of RIS and STAR-IOS to receiver vehicles which are located in transmission and reflection areas. In addition, we supposed that the fading channels between RIS and STAR-IOS followed Fisher-Snedecor $\mathcal{F}$ distribution. For the considered system model, we first introduced marginal distributions for the received SNR at receiver vehicles, and then, we obtained the closed-form expression of the OP, the upper bound of EC, and EE. Eventually, our numerical results showed that considering the RIS/STAR-IOS is beneficial for vehicular communications and it can remarkably improve the performance of ITS. Additionally, the results revealed that the NOMA scheme can provide a lower OP and higher values of EC and EE compared with the OMA scenario. 
\bibliographystyle{IEEEtran}
\bibliography{sample.bib}

\begin{thebibliography}{10}
\providecommand{\url}[1]{#1}
\csname url@samestyle\endcsname
\providecommand{\newblock}{\relax}
\providecommand{\bibinfo}[2]{#2}
\providecommand{\BIBentrySTDinterwordspacing}{\spaceskip=0pt\relax}
\providecommand{\BIBentryALTinterwordstretchfactor}{4}
\providecommand{\BIBentryALTinterwordspacing}{\spaceskip=\fontdimen2\font plus
\BIBentryALTinterwordstretchfactor\fontdimen3\font minus
  \fontdimen4\font\relax}
\providecommand{\BIBforeignlanguage}[2]{{%
\expandafter\ifx\csname l@#1\endcsname\relax
\typeout{** WARNING: IEEEtran.bst: No hyphenation pattern has been}%
\typeout{** loaded for the language `#1'. Using the pattern for}%
\typeout{** the default language instead.}%
\else
\language=\csname l@#1\endcsname
\fi
#2}}
\providecommand{\BIBdecl}{\relax}
\BIBdecl

\bibitem{wang2021green}
J.~Wang, K.~Zhu, and E.~Hossain, ``{Green Internet of Vehicles (IoV) in the 6G
  era: Toward sustainable vehicular communications and networking},''
  \emph{IEEE Trans. Green Commun. Netw}, vol.~6, no.~1, pp. 391--423, 2021.

\bibitem{noor20226g}
M.~Noor-A-Rahim, Z.~Liu, H.~Lee, M.~O. Khyam, J.~He, D.~Pesch, K.~Moessner,
  W.~Saad, and H.~V. Poor, ``{6G for vehicle-to-everything (V2X)
  communications: Enabling technologies, challenges, and opportunities},''
  \emph{Proc. IEEE}, vol. 110, no.~6, pp. 712--734, 2022.

\bibitem{basar2019wireless}
E.~Basar, M.~Di~Renzo, J.~De~Rosny, M.~Debbah, M.-S. Alouini, and R.~Zhang,
  ``{Wireless communications through reconfigurable intelligent surfaces},''
  \emph{IEEE access}, vol.~7, pp. 116\,753--116\,773, 2019.

\bibitem{liu2021star}
Y.~Liu, X.~Mu, J.~Xu, R.~Schober, Y.~Hao, H.~V. Poor, and L.~Hanzo, ``{STAR:
  Simultaneous transmission and reflection for 360° coverage by intelligent
  surfaces},'' \emph{IEEE Wirel. Commun.}, vol.~28, no.~6, pp. 102--109, 2021.

\bibitem{saito2013non}
Y.~Saito, Y.~Kishiyama, A.~Benjebbour, T.~Nakamura, A.~Li, and K.~Higuchi,
  ``{Non-orthogonal multiple access (NOMA) for cellular future radio access},''
  in \emph{2013 IEEE 77th Veh. Technol. Conf. (VTC Spring)}.\hskip 1em plus
  0.5em minus 0.4em\relax IEEE, 2013, pp. 1--5.

\bibitem{chen2020resource}
Y.~Chen, Y.~Wang, J.~Zhang, and Z.~Li, ``{Resource allocation for intelligent
  reflecting surface aided vehicular communications},'' \emph{IEEE Trans. Veh.
  Technol.}, vol.~69, no.~10, pp. 12\,321--12\,326, 2020.

\bibitem{singh2022visible}
G.~Singh, A.~Srivastava, and V.~A. Bohara, ``{Visible light and reconfigurable
  intelligent surfaces for beyond 5G V2X communication networks at road
  intersections},'' \emph{IEEE Trans. Veh. Technol.}, vol.~71, no.~8, pp.
  8137--8151, 2022.

\bibitem{wang2020outage}
J.~Wang, W.~Zhang, X.~Bao, T.~Song, and C.~Pan, ``{Outage analysis for
  intelligent reflecting surface assisted vehicular communication networks},''
  in \emph{GLOBECOM 2020-2020 IEEE Glob. Commun. Conf.}\hskip 1em plus 0.5em
  minus 0.4em\relax IEEE, 2020, pp. 1--6.

\bibitem{gu2022intelligent}
X.~Gu, G.~Zhang, Y.~Ji, W.~Duan, M.~Wen, Z.~Ding, and P.-H. Ho, ``{Intelligent
  surface aided D2D-V2X system for low-latency and high-reliability
  communications},'' \emph{IEEE Trans. Veh. Technol.}, vol.~71, no.~11, pp.
  11\,624--11\,636, 2022.

\bibitem{pala2023design}
S.~Pala, P.~Saikia, S.~K. Singh, K.~Singh, and C.-P. Li, ``{Design of
  RIS-assisted Full Duplex 6G-V2X Communications},'' in \emph{2023 IEEE Wirel.
  Commun. and Netw. Conf. (WCNC)}.\hskip 1em plus 0.5em minus 0.4em\relax IEEE,
  2023, pp. 1--6.

\bibitem{ozcan2021reconfigurable}
Y.~U. Ozcan, O.~Ozdemir, and G.~K. Kurt, ``{Reconfigurable intelligent surfaces
  for the connectivity of autonomous vehicles},'' \emph{IEEE Trans. Veh.
  Technol.}, vol.~70, no.~3, pp. 2508--2513, 2021.

\bibitem{ai2021secure}
Y.~Ai, A.~Felipe, L.~Kong, M.~Cheffena, S.~Chatzinotas, and B.~Ottersten,
  ``{Secure vehicular communications through reconfigurable intelligent
  surfaces},'' \emph{IEEE Trans. Veh. Technol.}, vol.~70, no.~7, pp.
  7272--7276, 2021.

\bibitem{shaikh2022performance}
M.~H.~N. Shaikh, K.~Rabie, X.~Li, T.~Tsiftsis, and G.~Nauryzbayev, ``{On the
  Performance of Dual RIS-assisted V2I Communication under Nakagami-m
  Fading},'' in \emph{2022 IEEE 96th Veh. Technol. Conf. (VTC2022-Fall)}.\hskip
  1em plus 0.5em minus 0.4em\relax IEEE, 2022, pp. 1--5.

\bibitem{guo2023star}
K.~Guo, R.~Liu, M.~Alazab, R.~H. Jhaveri, X.~Li, and M.~Zhu,
  ``{STAR-RIS-Empowered Cognitive Non-Terrestrial Vehicle Network with NOMA},''
  \emph{IEEE Trans. Intell. Veh.}, 2023.

\bibitem{bjornson2019intelligent}
E.~Bj{\"o}rnson, {\"O}.~{\"O}zdogan, and E.~G. Larsson, ``{Intelligent
  reflecting surface versus decode-and-forward: How large surfaces are needed
  to beat relaying?}'' \emph{IEEE Wireless Commun. Lett.}, vol.~9, no.~2, pp.
  244--248, 2019.

\end{thebibliography}
\end{document}